\documentclass{article}

\usepackage{amsmath}
\usepackage{url}
\usepackage{amssymb}
\usepackage{amsthm}
\usepackage{hyperref}

\newtheorem{prop}{Proposition}
\newtheorem{conj}{Conjecture}

\title{Weight spectra of some families of GRM codes}
\author{Minjia Shi, Zhaokang Xing, and Patrick Sol\'e%
	\thanks{Corresponding author: Minjia Shi.}%
	\thanks{Minjia Shi and Zhaokang Xing are with the Key Laboratory of Intelligent Computing Signal
		Processing, Ministry of Education, School of Mathematical Sciences, Anhui
		University, Hefei 230601, China, and also with the State Key Laboratory of
		Integrated Service Networks, Xidian University, Xi'an 710071, China
		(e-mail: smjwcl.good@163.com; 19565753226@163.com). Patrick Sol\'e is with
		I2M (Aix Marseille Univ, CNRS), Marseille, France
		(e-mail: sole@enst.fr).}}
\date{}

\begin{document}
\maketitle
\begin{abstract}
Let $\mathrm{RM}_q(r,m)$ denote the generalized Reed--Muller code of order $r$ and length $q^m$ over the finite field $\mathbb{F}_q$.
We completely determine the weight spectra of $\mathrm{RM}_3(2m-i,m)$ for $i=3,4$, $\mathrm{RM}_4(3m-i,m)$ for $i=2,3$, $\mathrm{RM}_5(4m-4,m)$, and $\mathrm{RM}_7(6m-4,m)$, in the ranges of $m$ specified in the corresponding results.
The proofs proceed by induction on $m$, using a sumset inclusion for GRM weight spectra together with results on low-weight codewords. The required base cases are established by computations in Magma and, for certain ternary cases, exact constraint solving with Z3.
\end{abstract}
\section{Introduction}
 
 Generalized Reed--Muller (GRM) codes were first introduced by Kasami, Lin, and Peterson~\cite{Kasami} and by Weldon~\cite{Weldon}, and were subsequently studied by Delsarte, Goethals, and MacWilliams~\cite{Delsarte1970}.
GRM codes constitute an important family of evaluation codes with applications in various areas of information theory, such as power control for OFDM communication systems \cite{Paterson2000}, the construction of quantum error-correcting codes \cite{GalindoHernandoRuano2015}, and capacity-achieving coding over certain classes of non-binary symmetric channels \cite{ReevesPfister2023}. Determining the Hamming weights of GRM codes has attracted considerable attention. We summarize the known results below.

For GRM codes of order $0$, the nonzero entries of the weight distribution are $A_0=1$ and $A_{q^m}=q-1$. For order $1$, they are
$A_0=1, A_{(q-1)q^{m-1}}=q(q^m-1), A_{q^m}=q-1.$
The weight distribution for order $2$ was investigated by McEliece~\cite{McEliece}; Li~\cite{Li2019} subsequently corrected errors in the general $q$-ary case and provided a complete account. By the duality relation and the MacWilliams identity, the weight distributions for orders $m(q-1)-1$, $m(q-1)-2$, and $m(q-1)-3$ can be obtained from those for orders $0$, $1$, and $2$, respectively, whenever these orders lie in the admissible range. For order $m(q-1)$, the code is the whole space $\mathbb{F}_q^{q^m}$, and hence
$A_i=\binom{q^m}{i}(q-1)^i, 0\leq i\leq q^m.$
Determining the complete weight distribution for intermediate orders is generally much more difficult. A related problem is to determine the \emph{weight spectrum}, namely the set of distinct weights that occur; throughout this paper, this set includes the zero weight.

For $q=2$, the weight spectra of $\mathrm{RM}_2(m-i,m)$ for $i=1,2,3,4,5,6$ have been determined in the parameter ranges treated in \cite{ShiLiNeriSole2019,CarletSole2023,Carlet2024,LouWang2025}. For $q=3$, the spectra of $\mathrm{RM}_3(2m-i,m)$ for $i=1,2$ were determined in \cite{ShiLiNeriSole2019}. The same reference gives the spectrum of $\mathrm{RM}_4(3m-1,m)$. For $q=5$, spectra for $\mathrm{RM}_5(4m-i,m)$ with $i=1,2,3,5$ are treated in \cite{ShiLiNeriSole2019,Golalizadeh2026}. For $q=7$, results for $\mathrm{RM}_7(6m-i,m)$ with $i=1,2,3$ are discussed in \cite{Golalizadeh2026}; that work also gives an almost complete determination for $i=9$, with weights $32$ and $33$ remaining outside its mathematical proof.

The weight spectra of GRM codes nevertheless remain unknown for many parameters. In this paper, we completely determine the spectra of $\mathrm{RM}_3(2m-i,m)$ for $i=3,4$, $\mathrm{RM}_4(3m-i,m)$ for $i=2,3$, $\mathrm{RM}_5(4m-4,m)$, and $\mathrm{RM}_7(6m-4,m)$.
Our main tool is induction on $m$, combined with the sumset inclusion in \cite[Corollary~25]{ShiLiNeriSole2019}. Although some of these spectra can, in principle, be extracted from known weight distributions by the MacWilliams identity, explicit descriptions do not appear to have been given in the literature. We provide direct inductive proofs of these descriptions. The base cases use computations in Magma~\cite{magma} and, for the ternary family $\mathrm{RM}_3(2m-4,m)$, exact constraint computations implemented in Python with the Z3 solver~\cite{Z3}. The latter computations and their verification procedure are described in Section~\ref{sec:computational}.

The paper is organized as follows. Section~\ref{sec:preliminaries} introduces the notation and results used throughout the paper. Sections~\ref{sec:ternary}, \ref{sec:quaternary}, \ref{sec:quinary}, and~\ref{sec:seven} determine weight spectra over $\mathbb{F}_3$, $\mathbb{F}_4$, $\mathbb{F}_5$, and $\mathbb{F}_7$, respectively. Section~\ref{open problem} recapitulates the obtained results, and mentions some challenging open problems.



\section{Preliminaries}\label{sec:preliminaries}
Let $\mathbb{F}_q$ be the finite field with $q$ elements, where $q$ is a prime power.
A linear code of length $n$ over $\mathbb{F}_q$ is an $\mathbb{F}_q$-subspace of $\mathbb{F}_q^n$.
The dimension of the code is its vector-space dimension over $\mathbb{F}_q$.
The Hamming weight $w_H(\boldsymbol{x})$ of an element $\boldsymbol{x}=(x_1,\dots,x_n)\in\mathbb{F}_q^n$ is the number of indices $i$ such that $x_i\neq 0$.
The minimum distance of a nonzero linear code is the minimum Hamming weight of a nonzero codeword.

The weights of a code $C$ are the Hamming weights of all its codewords.
The set of distinct weights, including zero, is called the weight spectrum and is denoted by $\mathrm{wt}(C)$; thus
\[
\mathrm{wt}(C) = \bigl\{i \in \{0,1,\ldots,n\} \,\big|\, \exists\, \mathbf c \in C \text{ such that } w_H(\mathbf c) = i\bigr\}.
\]
For $0\leq i\leq n$, let $A_i$ denote the number of codewords of weight $i$. The sequence $(A_0,A_1,\ldots,A_n)$ is called the weight distribution. The first, second, and third weights refer to the first three distinct positive weights, when they exist.

Let $m$ be a positive integer, and consider $R_m:=\mathbb{F}_q[x_1,\dots,x_m]$, the ring of polynomials in $m$ variables over $\mathbb{F}_q$.
Moreover, list all the points of $\mathbb{F}_q^m$ as $\mathbf{P}_1,\dots,\mathbf{P}_n$, where $n=q^m$, and consider the evaluation map
\[
\mathrm{ev}_m\colon R_m\to \mathbb{F}_q^n,
f\mapsto \bigl(f(\mathbf{P}_1),\dots,f(\mathbf{P}_n)\bigr).
\]

Let $r$ be an integer such that $0\le r\le (q-1)m$.
The generalized Reed--Muller code (GRM code) of order $r$ in $m$ variables is defined by
\[
\mathrm{RM}_q(r,m):=\bigl\{\mathrm{ev}_m(f)\,\big|\, f\in R_m,\ \deg(f)\le r\bigr\},
\]
where $\deg(f)$ denotes the total degree of $f$, with the convention $\deg(0)=-\infty$. Since $\beta^q=\beta$ for every $\beta\in\mathbb{F}_q$, a factor $x_i^{q+d}$, with $d\geq0$, can be replaced by $x_i^{d+1}$ without changing the evaluations. Repeating these reductions yields a unique representative with degree at most $q-1$ in each variable. A polynomial satisfying these individual degree bounds is said to be \emph{reduced}. Since reduction never increases total degree, we also have
\[
\mathrm{RM}_q(r,m)=\bigl\{\mathrm{ev}_m(f)\bigm| f\in R_m,\; \deg(f)\leq r,
\ \deg_{x_i}(f)<q\text{ for }1\leq i\leq m\bigr\},
\]
where $\deg_{x_i}(f)$ denotes the degree of $f$ in $x_i$. We identify polynomial functions with their reduced representatives whenever degree arguments are used. In particular, a reduced polynomial that vanishes at every point of $\mathbb{F}_q^m$ is the zero polynomial.

The support of $f\in R_m$ is the set $\{\boldsymbol{x}\in\mathbb{F}_q^m: f(\boldsymbol{x})\neq 0\}$.
The cardinality of this support is denoted by $|f|$ and equals the Hamming weight of $\mathrm{ev}_m(f)$.
We use the notation $f\in \mathrm{RM}_q(r,m)$ to mean that
$(f(\mathbf{P}_1),\dots,f(\mathbf{P}_n))\in \mathrm{RM}_q(r,m)$.

An (affine) hyperplane of \(\mathbb{F}_q^m\) is a subset of the form
\[
H=\left\{\boldsymbol x=(x_1,\ldots,x_m)\in\mathbb{F}_q^m:
a_1x_1+\cdots+a_mx_m=b\right\},
\]
where \(a_1,\ldots,a_m,b\in\mathbb{F}_q\) and $(a_1,\ldots,a_m)\neq(0,\ldots,0)$.

Let $t$ be a positive integer. Given $t$ sets of integers $A_1,\dots,A_t$, we define the set of sums
\[
\bigoplus_{i=1}^t A_i:=\bigl\{a_1+a_2+\dots+a_t \,\big|\, a_i\in A_i\bigr\}.
\]

The following sumset inclusion is essential to the inductive proofs in the subsequent sections. We state it in the range where both codes have nonnegative order and a positive number of variables.

\begin{prop}[\cite{ShiLiNeriSole2019}]\label{prop:spectrum-inclusion}
	Let $r,m$ be integers such that $m\geq2$ and $0\leq r\leq(m-1)(q-1)-1$. Then
	\[
	\mathrm{wt}\big(\mathrm{RM}_q\big(m(q-1)-r-1,\,m\big)\big) \supseteq \bigoplus_{i=1}^q \mathrm{wt}\big(\mathrm{RM}_q\big((m-1)(q-1)-r-1,\,m-1\big)\big).
	\]
\end{prop}

The next two results give the minimum distance and the second weight of GRM codes, respectively. For the second weight, we record only the cases needed below.
	
\begin{prop}[\cite{ShiLiNeriSole2019}]\label{prop:minimum-distance}
    Let $m\geq1$ and $0\leq r\leq m(q-1)$, and write
    $r=Q(q-1)+S$, where $0\leq S\leq q-2$.
    The minimum distance of $\mathrm{RM}_q(r,m)$ is
    \[
    d_q(r,m)=(q-S)q^{m-Q-1}.
    \]
    In particular, $d_q(m(q-1),m)=1$.
\end{prop}

\begin{prop}[\cite{BalletRolland2014}]\label{prop:second-weight}
    Let $q\geq3$, $m\geq2$, and $r=s(q-1)+t$, where
    $0\leq s\leq m-2$ and $1\leq t\leq q-1$.
    In each of the following cases, the second weight of $\mathrm{RM}_q(r,m)$ is
    \[
    d_q(r,m)+c\,q^{m-s-2},
    \]
    where
    \[
    c=\begin{cases}
    q-1,&\text{if }q=3,\ s=m-2,\ t=1,\\
    t-1,&\text{if }1<t\leq(q+1)/2\text{ or }t=q-1.
    \end{cases}
    \]
\end{prop}

We also use the following geometric result on sets meeting every affine hyperplane.
\begin{prop}[\cite{BrouwerSchrijver1978}]\label{prop:intersect}
	The minimum cardinality of a subset of $\mathbb{F}_q^m$ which intersects all hyperplanes is $m(q-1)+1$.
\end{prop}

The following factorization property will be used to analyze codewords whose support avoids a hyperplane.
\begin{prop}[\cite{Delsarte1970}]\label{prop:fx=0}
    Let $f\in R_m$ be reduced and let $a\in\mathbb{F}_q$.
    If $f(\boldsymbol{x})=0$ whenever $x_1=a$, then
    $f(\boldsymbol{x})=(x_1-a)g(\boldsymbol{x})$ for some reduced polynomial $g$ such that the product on the right is reduced. If $f\neq0$, then $\deg(g)\leq\deg(f)-1$ and $\deg_{x_1}(g)\leq q-2$.
\end{prop}

The next two results describe minimum-weight codewords of GRM codes and second-weight codewords in the ternary case needed below.

\begin{prop}[\cite{Leducq2012}]\label{prop:minimun weight codeword}
	Let $r=t(q-1)+s$, where $0\leq t\leq m-1$ and $0\leq s\leq q-2$. Every minimum-weight codeword of $\mathrm{RM}_q(r,m)$ is equivalent, under an invertible affine change of variables, to the evaluation of a function of the form
	\[
	\forall \boldsymbol{x}\in\mathbb{F}_q^m,
	f(\boldsymbol{x})=c\prod_{i=1}^{t}\bigl(x_i^{q-1}-1\bigr)\prod_{j=1}^{s}(x_{t+1}-b_j)
	\]
	where $c\in\mathbb{F}_q^*$ and $b_1,\ldots,b_s$ are distinct elements of $\mathbb{F}_q$.
\end{prop}

\begin{prop}[\cite{Leducq2013}]\label{prop:second weight codeword}
	Let $m\geq3$ and $1\leq t\leq m-2$. Up to an invertible affine change of variables, every second-weight codeword of $\mathrm{RM}_3(2t+1,m)$ is the evaluation of a function of the form
	\[
	\forall \boldsymbol{x}\in\mathbb{F}_3^m,
	f(\boldsymbol{x})=\alpha\prod_{i=1}^{t-1}\bigl(1-x_i^2\bigr)\,x_t x_{t+1} x_{t+2},
	\]
	where $\alpha\in\mathbb{F}_3^*$.
\end{prop}
\section{Ternary GRM codes}\label{sec:ternary}
\subsection{The weight spectrum of \texorpdfstring{$\mathrm{RM}_3(2m-3,m)$}{RM3(2m-3,m)}}
    In \cite{ShiLiNeriSole2019}, the authors showed that the set consisting of $0, 6$, and all integers from $8$ to $3^m$ is contained in the weight spectrum of $\mathrm{RM}_3(2m-3,\,m)$. We now prove that no other weights occur.
\begin{prop}\label{prop:rm3-spectrum}
	If $m \ge 3$, then
	$\mathrm{wt}\big(\mathrm{RM}_3(2m-3,\,m)\big) = \big\{0,\,6,\,8,\,9,\,\dots,\,3^m\big\}$.
\end{prop}
\begin{proof}
	We proceed by induction on $m$. A computation in Magma~\cite{magma} gives
	\[
\mathrm{wt}\big(\mathrm{RM}_3(3,3)\big) = \big\{0,\,6,\,8,\,9,\,\dots,\,27\big\}.
\]
	Thus the assertion holds for $m=3$.
	
	Now assume that the assertion holds for some $m\ge 3$. By Proposition~\ref{prop:spectrum-inclusion},
	\[
	\mathrm{wt}\big(\mathrm{RM}_3(2m-1,\,m+1)\big) \supseteq \bigoplus_{i=1}^3 \big\{0,\,6,\,8,\,9,\,\dots,\,3^m\big\}.
	\]
	Let
    \[
    A=\bigoplus_{i=1}^3 \big\{0,\,6,\,8,\,9,\,\dots,\,3^m\big\}.
    \]
    We claim that
    \[
    A=\big\{0,\,6,\,8,\,9,\,\dots,\,3^{m+1}\big\}.
    \]
	Indeed, the following decompositions cover all the required weights.
\begin{itemize}
\item If $x = 0$, $x=6$, or $8 \leq x \leq 3^m$, we may write $x = x+0+0$, so $x \in A$.
	\item If $x = 3^m + a$ with $1 \leq a \leq 7$, we may write $x = 8 + \big(3^m + a - 8\big) + 0$, which shows that $x\in A$.
\item	If $3^m+8 \leq x \leq 2\cdot 3^m$, we may write $x = 3^m + (x-3^m) + 0$, which shows that $x\in A$.
\item	If $x = 2\cdot 3^m + a$ with $1 \leq a \leq 7$, we may write $x = 3^m + 8 + \big(3^m + a - 8\big)$, which shows that $x\in A$.
\item	If $2\cdot 3^m+8 \leq x \leq 3^{m+1}$, we may write $x = 3^m + 3^m + \big(x - 2\cdot 3^m\big)$, which shows that $x\in A$.
\end{itemize}
    Here $8\leq 3^m+a-8\leq 3^m$ for $1\leq a\leq 7$ in the induction range, so every displayed summand belongs to the defining set of $A$.
	Therefore,
	$\big\{0,6,8,9,\dots,3^{m+1}\big\} \subseteq A$.
	Conversely, every element of $A$ is at most $3\cdot 3^m = 3^{m+1}$. Moreover, since the smallest positive element of the summand set $\big\{0,\,6,\,8,\,9,\,\dots,\,3^m\big\}$ is $6$ and every other positive element is at least $8$, none of $1,2,3,4,5,7$ can belong to $A$. Hence $A \subseteq \big\{0,6,8,9,\dots,3^{m+1}\big\}$.
	As a consequence,
	$A = \big\{0,6,8,9,\dots,3^{m+1}\big\}$.
	
	From Proposition~\ref{prop:minimum-distance}, we have $d_3(2m-1,\,m+1)=6$, and from Proposition~\ref{prop:second-weight}, the second weight of $\mathrm{RM}_3(2m-1,\,m+1)$ is $8$. Hence the weights $1, 2, 3, 4, 5$ and $7$ cannot occur. Since the code has length $3^{m+1}$, no weight can exceed $3^{m+1}$. Together with the above inclusion, this yields
	\[
	\mathrm{wt}\big(\mathrm{RM}_3(2m-1,\,m+1)\big) = \big\{0,\,6,\,8,\,9,\,\dots,\,3^{m+1}\big\}.
	\]
	This completes the proof.
\end{proof}

\subsection{The weight spectrum of \texorpdfstring{$\mathrm{RM}_3(2m-4,m)$}{RM3(2m-4,m)}}
The spectrum of $\mathrm{RM}_3(2,3)$ is
$\{0,9,12,15,18,21,24,27\}$; see \cite{Vance2000}.
Together with Proposition~\ref{prop:spectrum-inclusion}, this gives the following inclusion.
\begin{prop}\label{prop:ternary-multiples}
    If $m\geq3$, then
    \[
    \mathrm{wt}\bigl(\mathrm{RM}_3(2m-4,m)\bigr)\supseteq\{0,9,12,15,18,\ldots,3^m\}.
    \]
\end{prop}
\begin{proof}
    The case $m=3$ follows from the spectrum above. Assume that the inclusion holds for some $m\geq3$, and put
    $N=3^{m-1}$ and $B=\{0,3,4,\ldots,N\}$, so that the asserted weights at level $m$ form $3B=\{3b:b\in B\}$.
    Since $N\geq9$, the sum of three copies of $B$ is $\{0,3,4,\ldots,3N\}$.
    Indeed, weights up to $N$ use one nonzero summand; $N+1$ and $N+2$ are $3+(N-2)$ and $3+(N-1)$; weights from $N+3$ to $2N$ use $N+(x-N)$.
    The remaining weights are obtained by adding $N$ to these decompositions, using a third summand when necessary. No sum is $1$ or $2$.
    Multiplying by $3$ and applying Proposition~\ref{prop:spectrum-inclusion} proves the inclusion at level $m+1$.
\end{proof}
By Propositions~\ref{prop:minimum-distance} and~\ref{prop:second-weight}, the first two positive weights of $\mathrm{RM}_3(2m-4,m)$ are $9$ and $12$, respectively.
To determine the spectrum for $m=4$, it therefore remains to decide which integers $w$ with $13\leq w\leq80$ and $3\nmid w$ occur. We first exclude weights $13$ and $14$ and then establish the occurrence of all larger admissible weights.

\begin{prop}\label{prop:no 13,14}
	If $m\geq6$, then $13,14\notin\mathrm{wt}(\mathrm{RM}_3(2m-4,m))$.
\end{prop}
\begin{proof}
	We proceed by induction on $m$. For $m=6$, an exact constraint computation implemented
	in Python using Z3~\cite{Z3} verifies that
	$13,14\notin\mathrm{wt}(\mathrm{RM}_3(8,6))$. Now assume that the assertion holds for some $m-1\ge6$.
	
	Let $f\in\mathrm{RM}_3(2m-4,m)$ satisfy $|f|\in\{13,14\}$, and let $S$ be its support. In this induction step $m\geq7$, so $|S|\leq14<2m+1$. Proposition~\ref{prop:intersect} therefore yields an affine hyperplane $H_0$ disjoint from $S$. Since invertible affine substitutions preserve GRM codes, we may assume that $H_0=\{\boldsymbol{x}:x_1=0\}$.
    By Proposition~\ref{prop:fx=0}, the reduced polynomial $f$ has the form
    \[
    f(\boldsymbol{x})=x_1g_1(\bar{\boldsymbol{x}})+x_1^2g_2(\bar{\boldsymbol{x}}),
    \bar{\boldsymbol{x}}=(x_2,\ldots,x_m),
    \]
    where $\deg(g_1)\leq2m-5$ and $\deg(g_2)\leq2m-6$.
    On the two remaining parallel hyperplanes $H_1=\{\boldsymbol{x}:x_1=1\}$ and $H_2=\{\boldsymbol{x}:x_1=-1\}$, the restrictions of $f$ are
    \[
    u(\bar{\boldsymbol{x}})=g_1(\bar{\boldsymbol{x}})+g_2(\bar{\boldsymbol{x}}),
    v(\bar{\boldsymbol{x}})=-g_1(\bar{\boldsymbol{x}})+g_2(\bar{\boldsymbol{x}}),
    \]
    respectively. Thus $|f|=|u|+|v|$, and $u,v\in\mathrm{RM}_3(2m-5,m-1)$. By Proposition~\ref{prop:rm3-spectrum},
    \[
    \mathrm{wt}\bigl(\mathrm{RM}_3(2m-5,m-1)\bigr)
    =\{0,6,8,9,\ldots,3^{m-1}\}.
    \]

    \noindent\textbf{Case 1: $|u|\neq0$ and $|v|\neq0$.}
	
	If $|f|=|u|+|v|=13$, then no such pair of nonzero weights exists. If $|f|=|u|+|v|=14$, then $(|u|,|v|)=(6,8)$ or $(8,6)$. We first consider $|u|=6$ and $|v|=8$. Put $n=m-1$ and $d=2m-5=2n-3$. Since the minimum distance of $\mathrm{RM}_3(d-1,n)$ is $9$, neither $u$ nor $v$ belongs to this code. Thus $\deg u=\deg v=d$. Moreover, $u+v=2g_2$ has degree at most $d-1$, so $u_d=-v_d$, where $h_d$ denotes the homogeneous component of degree $d$ of a reduced polynomial $h$.
	
	To show that this equality is impossible, let $\overline{p}$ denote the reduced representative of a polynomial function $p$ on $\mathbb{F}_3^n$, obtained using $y_i^3=y_i$. For a reduced polynomial $h$ of degree $d$, define the vector space
	\[
	K(h)=\left\{L\in\operatorname{span}_{\mathbb{F}_3}\{y_1,\ldots,y_n\}:\deg\overline{Lh}\leq d\right\}.
	\]
	Since multiplication by $L$ increases degree by at most one and reduction never increases degree, $K(h)$ depends only on $h_d$. In particular, $u_d=-v_d$ implies $K(u)=K(v)$.
	
	Furthermore, $\dim K(h)$ is invariant under invertible affine changes of variables. Indeed, let $T(\boldsymbol{y})=B\boldsymbol{y}+\boldsymbol{b}$ with $B$ invertible. Affine substitution followed by reduction preserves the degree filtration, and its inverse does so as well. Hence, for every homogeneous linear form $L$,
	\[
	\deg\overline{Lh}\leq d
	\Longleftrightarrow
	\deg\overline{(L\circ T)(h\circ T)}\leq d.
	\]
	The constant term of $L\circ T$ contributes a polynomial of degree at most $d$, so it does not affect this condition. Thus $L\mapsto L(B\boldsymbol{y})$ gives an isomorphism from $K(h)$ to $K(\overline{h\circ T})$.
	
	By Proposition~\ref{prop:minimun weight codeword}, after an affine change of variables the degree-$d$ component of $u$ is
	\[
	c\,y_1^2\cdots y_{n-2}^2y_{n-1}, c\in\mathbb{F}_3^*.
	\]
	Multiplication by any of $y_1,\ldots,y_{n-2}$ produces a cubic factor, whose reduction lowers the degree from $d+1$ to $d-1$. Multiplication by $y_{n-1}$ or $y_n$ instead produces two distinct reduced monomials of degree $d+1$, so their coefficients cannot cancel. Therefore, in these coordinates,
	\[
	K(u)=\operatorname{span}_{\mathbb{F}_3}\{y_1,\ldots,y_{n-2}\},
	\dim K(u)=n-2.
	\]
	By Proposition~\ref{prop:second weight codeword}, after a possibly different affine change of variables the degree-$d$ component of $v$ is
	\[
	\alpha\,y_1^2\cdots y_{n-3}^2y_{n-2}y_{n-1}y_n,
    \alpha\in\mathbb{F}_3^*.
	\]
	Here multiplication by $y_1,\ldots,y_{n-3}$ lowers the degree after reduction, whereas multiplication by $y_{n-2},y_{n-1},y_n$ produces three distinct reduced monomials of degree $d+1$. Consequently, in these coordinates,
	\[
	K(v)=\operatorname{span}_{\mathbb{F}_3}\{y_1,\ldots,y_{n-3}\},
	\dim K(v)=n-3.
	\]
	By affine invariance, these dimensions also hold in the original coordinates, contradicting $K(u)=K(v)$. The case $(|u|,|v|)=(8,6)$ follows by interchanging $u$ and $v$. Therefore, $|f|$ cannot be $13$ or $14$ in Case 1.
	
	\noindent\textbf{Case 2: exactly one of $|u|$ and $|v|$ is zero.}
	
	Suppose first that $|v|=0$. Since $v$ is reduced, it is the zero polynomial, and hence $g_1=g_2$. Therefore $u=2g_2$ has degree at most $2m-6$, so
    \[
    u\in\mathrm{RM}_3(2m-6,m-1)=\mathrm{RM}_3(2(m-1)-4,m-1).
    \]
    This contradicts the induction hypothesis because $|u|=|f|\in\{13,14\}$. If $|u|=0$, then $g_1=-g_2$ and $v=2g_2$, giving the same contradiction. The two cases exhaust all possibilities, completing the induction.
\end{proof}

\begin{prop}\label{prop:15}
	For every $m\geq3$, the third weight of $\mathrm{RM}_3(2m-4,m)$ is $15$.
\end{prop}
\begin{proof}
	For $m=3$, the assertion follows from the spectrum of $\mathrm{RM}_3(2,3)$ given above. Exact constraint computations implemented in Python with Z3~\cite{Z3} exclude weights $13$ and $14$ from both $\mathrm{RM}_3(4,4)$ and $\mathrm{RM}_3(6,5)$; see Section~\ref{sec:computational}. Proposition~\ref{prop:no 13,14} excludes these weights for all $m\geq6$. Since the first two positive weights are $9$ and $12$, and weight $15$ occurs by Proposition~\ref{prop:ternary-multiples}, it follows that the third weight is $15$ for every $m\geq3$.
\end{proof}

\begin{prop}
		If $m\ge 4$, then
    \[
    \mathrm{wt}\big(\mathrm{RM}_3(2m-4,m)\big)=\{0,\,9,\,12,\,15,\,16,\,17,\,\dots,\,3^m\}.
    \]
\end{prop}
\begin{proof}
	We proceed by induction on $m$. An exact constraint computation implemented in Python using Z3~\cite{Z3} verifies that
	\[
\mathrm{wt}\big(\mathrm{RM}_3(4,4)\big) = \big\{0,\,9,\,12,\,15,\,16,\,17,\,\dots,\,81\}.
\]
	Thus the assertion holds for $m=4$.
	
	Now assume that the assertion holds for some $m\ge 4$. By Proposition~\ref{prop:spectrum-inclusion},
	\[
	\mathrm{wt}\big(\mathrm{RM}_3(2m-2,\,m+1)\big) \supseteq \bigoplus_{i=1}^3 \big\{0,\,9,\,12,\,15,\,16,\,17,\,\dots,\,3^m\}.
	\]
	Let
    \[
    A=\bigoplus_{i=1}^3 \big\{0,\,9,\,12,\,15,\,16,\,17,\,\dots,\,3^m\}.
    \]
    We claim that
    \[
    A=\big\{0,\,9,\,12,\,15,\,16,\,17,\,\dots,\,3^{m+1}\}.
    \]
	Indeed, the following decompositions cover all the required weights.
\begin{itemize}
\item If \(x=0\), \(x=9\), \(x=12\), or \(15\leq x\leq 3^m\), we may write $x=x+0+0$,
	so \(x\in A\). 
\item If \(x=3^m+a\) with \(1\leq a\leq 14\), we may write $x=15+\bigl(3^m+a-15\bigr)+0$, which shows that \(x\in A\). 

\item If $3^m+15\leq x\leq 2\cdot3^m$,
	we may write $x=3^m+\bigl(x-3^m\bigr)+0$, which shows that \(x\in A\).
\item	If \(x=2\cdot3^m+a\) with \(1\leq a\leq14\), we may write $x=3^m+15+\bigl(3^m+a-15\bigr)$, which shows that \(x\in A\).

\item   If $2\cdot3^m+15\leq x\leq3^{m+1}$,
	we may write $x=3^m+3^m+\bigl(x-2\cdot3^m\bigr)$, which shows that \(x\in A\). 
    \end{itemize}
    Here $15\leq 3^m+a-15\leq 3^m$ for $1\leq a\leq 14$ in the induction range, so every displayed summand belongs to the defining set of $A$.
    Therefore,
	\[
    \{0,9,12,15,16,17,\ldots,3^{m+1}\}\subseteq A.
    \]
	Conversely, every element of \(A\) is at most $3\cdot3^m=3^{m+1}$.
	Moreover, every positive element of $\{0,9,12,15,16,17,\ldots,3^m\}$
	is either \(9\), \(12\), or at least \(15\). Hence a sum of three such elements can never be any of
	$1,2,3,4,5,6,7,8,10,11,13,14$. Thus \[
    A\subseteq\{0,9,12,15,16,17,\ldots,3^{m+1}\}.
    \]
	Consequently, \[
    A=\{0,9,12,15,16,17,\ldots,3^{m+1}\}.
    \]
	
	From Proposition~\ref{prop:minimum-distance}, we have $d_3(2m-2,\,m+1)=9$. From Proposition~\ref{prop:second-weight}, the second weight of $\mathrm{RM}_3(2m-2,\,m+1)$ is $12$. From Proposition~\ref{prop:15}, the third weight of $\mathrm{RM}_3(2m-2,\,m+1)$ is $15$. Hence the weights $1, 2, 3, 4, 5, 6, 7, 8$ and $10, 11, 13, 14$ cannot occur. Since the code has length $3^{m+1}$, no weight can exceed $3^{m+1}$. Together with the above inclusion, this yields
	\[
	\mathrm{wt}\big(\mathrm{RM}_3(2m-2,\,m+1)\big) = \big\{0,\,9,\,12,\,15,\,16,\,17,\,\dots,\,3^{m+1}\}.
	\]
	This completes the proof.
\end{proof}

\subsection{Computational verification}\label{sec:computational}
The ternary computations were implemented in Python using
Z3~\cite{Z3}, a satisfiability modulo theories (SMT) solver.
SMT asks whether a collection of logical constraints admits an
assignment satisfying all of them, with operations interpreted
according to specified mathematical theories, such as integer
arithmetic or fixed-width bit-vector arithmetic. Here, the
existence of a codeword of a prescribed weight is formulated
as such a constraint problem.

Let $C_m=\mathrm{RM}_3(2m-4,m), m\geq 3.$
The duality relation
$C_m^\perp=\mathrm{RM}_3(3,m)$
provides an exact membership test. For a prescribed positive
weight $w$, we seek pairwise distinct points
$P_1,\ldots,P_w\in\mathbb{F}_3^m$ and nonzero values
$a_1,\ldots,a_w\in\mathbb{F}_3^\ast$ satisfying
$\sum_{i=1}^{w} a_i P_i^\alpha=0$
in $\mathbb{F}_3$
for every $\alpha\in\{0,1,2\}^m
\text{ with }|\alpha|\leq 3,$
where
$
|\alpha|=\sum_{j=1}^{m}\alpha_j,
P_i^\alpha=\prod_{j=1}^{m}(P_i)_j^{\alpha_j},
$
and a factor with exponent zero is interpreted as $1$.
The evaluation vectors of these monomials span
$\mathrm{RM}_3(3,m)$. Hence the displayed equations are
equivalent to orthogonality to every vector in $C_m^\perp$.
Since $(C_m^\perp)^\perp=C_m$, they are necessary and sufficient
for the vector taking value $a_i$ at $P_i$ and zero elsewhere
to belong to $C_m$. Pairwise distinctness of the points and
nonvanishing of the values ensure that its weight is exactly $w$.

The search is divided according to the affine dimension $d$
of the support, namely the dimension of the smallest affine
subspace containing it. For a support of size $w$, necessarily
$
\lceil\log_3 w\rceil
\leq d\leq \min\{m,w-1\}.
$
Indeed, an affine subspace of dimension $d$ contains $3^d$
points, and $w$ points span an affine subspace of dimension
at most $w-1$. A support of affine dimension $d$ contains
$d+1$ affinely independent points. Since invertible affine
changes of coordinates preserve GRM codes and Hamming
weights, these points may be mapped to
$
0,\boldsymbol{e}_1,\ldots,\boldsymbol{e}_d,
$
with the entire support contained in
$\mathbb{F}_3^d\times\{0\}^{m-d}$, where
$\boldsymbol{e}_j$ is the $j$th standard basis vector.
Multiplying the codeword by a nonzero scalar then makes
its value at $0$ equal to $1$. Only the remaining support
points are sorted, with their values relabeled accordingly.
These normalizations remove redundant descriptions without
excluding any possible weight.

The constraints are encoded using bit vectors, which are
fixed-length strings of binary digits representing bounded
integers. Field elements are represented by $0,1,2$, and
the encoding implements arithmetic modulo $3$ explicitly.
Bit-vector arithmetic itself wraps around modulo a power
of $2$, so sufficient widths are required for all intermediate
products and sums before reduction modulo $3$.
A time limit of $3600$ seconds was imposed on each
weight--dimension pair.

A solver response \texttt{SAT} means that the encoded
constraints admit a satisfying assignment, that is, a choice
of the points and values satisfying all constraints.
Such an assignment supplies a candidate witness: an explicit
codeword intended to establish the occurrence of the requested
weight. Each candidate is checked independently of the
bit-vector encoding using integer arithmetic with reduction
modulo $3$. The checks verify that the support points are
distinct, that their assigned values are nonzero, and that
all dual constraints hold. The corresponding reduced
polynomial is also reconstructed by interpolation, and its
total degree, evaluations, and Hamming weight are verified.
In particular, this polynomial can be written explicitly as
\[
f(x_1,\ldots,x_m)
=
\sum_{i=1}^{w} a_i
\prod_{j=1}^{m}
\bigl(1-(x_j-(P_i)_j)^2\bigr).
\]
Each product is the indicator function of $P_i$, so $f$
has the prescribed evaluations and is reduced.
Checking that $\deg(f)\leq 2m-4$ therefore gives a direct
verification of membership in $C_m$.

A response \texttt{UNSAT} means that the encoded constraints
have no satisfying assignment. A weight is excluded only
when every possible affine dimension returns \texttt{UNSAT}.
A response \texttt{UNKNOWN}, or termination without a decision
because of the time limit, is inconclusive. The independent
witness checks certify the reported existence results;
the nonexistence results rely on the completeness and
correctness of the encoding together with Z3's
unsatisfiability decisions. The retained records include
the scripts, constraint instances in SMT-LIB format
(a standard input language for SMT solvers), structured
result files in JSON format, and explicit witnesses
for the satisfiable cases.

For each $w\in\{13,14\}$ and each $m\in\{4,5,6\}$,
all cases $3\leq d\leq m$ returned \texttt{UNSAT}.
These dimensions exhaust the possibilities, since
an affine subspace of dimension at most $2$ contains
at most $3^2=9<13$ points. Thus the computations exclude
weights $13$ and $14$ from
$\mathrm{RM}_3(4,4)$, $\mathrm{RM}_3(6,5)$, and
$\mathrm{RM}_3(8,6)$.

For $C_4=\mathrm{RM}_3(4,4)$, independently verified
witnesses were obtained for all $43$ weights in
\[
\{w\in\mathbb{Z}:17\leq w\leq 80,\ 3\nmid w\}.
\]
Weight $16$ is realized directly by $x_1x_2x_3x_4$,
which has degree $4$ and is nonzero at exactly
$2^4=16$ points of $\mathbb{F}_3^4$.
Together with the previously established inclusion
\[
\{0,9,12\}
\cup
\{3k:k\in\mathbb{Z},\ 5\leq k\leq 27\}
\subseteq \mathrm{wt}(C_4),
\]
these witnesses establish the occurrence of every weight
in $\{0,9,12,15,16,\ldots,81\}$.
The minimum distance $9$ and second weight $12$ exclude
weights $1,\ldots,8,10,11$, while the above
unsatisfiability results exclude $13$ and $14$.
Consequently,
\[
\mathrm{wt}\bigl(\mathrm{RM}_3(4,4)\bigr)
=
\{0,9,12,15,16,\ldots,81\}.
\]

\section{Quaternary GRM codes}\label{sec:quaternary}
\subsection{The weight spectrum of \texorpdfstring{$\mathrm{RM}_4(3m-2,m)$}{RM4(3m-2,m)}}
\begin{prop}
	If $m\geq1$, then
    \[
    \mathrm{wt}\bigl(\mathrm{RM}_4(3m-2,m)\bigr)
    =\{0\}\cup\{w\in\mathbb{Z}:3\leq w\leq4^m\}.
    \]
\end{prop}
\begin{proof}
	We proceed by induction on $m$. A nonzero affine polynomial in one variable is either constant, with weight $4$, or has exactly one zero, with weight $3$. Hence $\mathrm{wt}(\mathrm{RM}_4(1,1))=\{0,3,4\}$. A computation in Magma~\cite{magma} gives
	\[
\mathrm{wt}\big(\mathrm{RM}_4(4,2)\big) = \big\{0,\,3,\,4,\,5,\,\dots,\,16\big\}.
\]
	Thus the assertion holds for $m=1,2$.
	
	Now assume that the assertion holds for some $m\ge 2$. By Proposition~\ref{prop:spectrum-inclusion},
	\[
	\mathrm{wt}\big(\mathrm{RM}_4(3m+1,\,m+1)\big) \supseteq \bigoplus_{i=1}^4 \big\{0,\,3,\,4,\,5,\,\dots,\,4^m\big\}.
	\]
	Let
    \[
    A=\bigoplus_{i=1}^4 \big\{0,\,3,\,4,\,5,\,\dots,\,4^m\big\}.
    \]
    We claim that
    \[
    A=\big\{0,\,3,\,4,\,5,\,\dots,\,4^{m+1}\big\}.
    \]
	Indeed, the following decompositions cover all the required weights.
\begin{itemize}
\item If $x=0$ or $3\leq x\leq 4^m$, we may write $x=x+0+0+0$, so $x\in A$.
\item If $x=4^m+a$ with $1\leq a\leq 2$, we may write $x=3+\bigl(4^m+a-3\bigr)+0+0$, which shows that $x\in A$.
\item If $4^m+3\leq x\leq 2\cdot4^m$, we may write $x=4^m+\bigl(x-4^m\bigr)+0+0$, which shows that $x\in A$.
\item If $x=2\cdot4^m+a$ with $1\leq a\leq 2$, we may write $x=4^m+3+\bigl(4^m+a-3\bigr)+0$, which shows that $x\in A$.
\item If $2\cdot4^m+3\leq x\leq 3\cdot4^m$, we may write $x=4^m+4^m+\bigl(x-2\cdot4^m\bigr)+0$, which shows that $x\in A$.
\item If $x=3\cdot4^m+a$ with $1\leq a\leq 2$, we may write $x=4^m+4^m+3+\bigl(4^m+a-3\bigr)$, which shows that $x\in A$.
\item If $3\cdot4^m+3\leq x\leq 4^{m+1}$, we may write $x=4^m+4^m+4^m+\bigl(x-3\cdot4^m\bigr)$, which shows that $x\in A$.
\end{itemize}
	Here $m\geq2$ ensures that $3\leq4^m+a-3\leq4^m$ whenever $1\leq a\leq2$, so all the summands above belong to the defining set of $A$.
	Therefore, $\{0,3,4,5,\ldots,4^{m+1}\}\subseteq A$.
	Conversely, every element of $A$ is at most $4\cdot4^m=4^{m+1}$. Moreover, every positive summand is at least $3$, so none of $1,2$ can belong to $A$. Thus $A\subseteq\{0,3,4,5,\ldots,4^{m+1}\}$.
	As a consequence, $A=\{0,3,4,5,\ldots,4^{m+1}\}$.
	
	From Proposition~\ref{prop:minimum-distance}, we have $d_4(3m+1,\,m+1)=3$. Hence the weights $1, 2$ cannot occur. Since the code has length $4^{m+1}$, no weight can exceed $4^{m+1}$. Together with the above inclusion, this yields
	\[
	\mathrm{wt}\big(\mathrm{RM}_4(3m+1,\,m+1)\big) = \big\{0,\,3,\,4,\,5,\,\dots,\,4^{m+1}\big\}.
	\]
	This completes the proof.
\end{proof}

\subsection{The weight spectrum of \texorpdfstring{$\mathrm{RM}_4(3m-3,m)$}{RM4(3m-3,m)}}
\begin{prop}
	If $m\ge2$, then
    \[
    \mathrm{wt}\big(\mathrm{RM}_4(3m-3,m)\big)=\{0,\,4,\,6,\,7,\,\dots,\,4^m\}.
    \]
\end{prop}
\begin{proof}
	We proceed by induction on $m$. A computation in Magma~\cite{magma} gives
	\[
\mathrm{wt}\big(\mathrm{RM}_4(3,2)\big) = \big\{0,\,4,\,6,\,7,\,\dots,\,16\big\}.
\]
	Thus the assertion holds for $m=2$.
	
	Now assume that the assertion holds for some $m\ge 2$. By Proposition~\ref{prop:spectrum-inclusion},
	\[
	\mathrm{wt}\big(\mathrm{RM}_4(3m,\,m+1)\big) \supseteq \bigoplus_{i=1}^4 \big\{0,\,4,\,6,\,7,\,\dots,\,4^m\big\}.
	\]
	Let
    \[
    A=\bigoplus_{i=1}^4 \big\{0,\,4,\,6,\,7,\,\dots,\,4^m\big\}.
    \]
    We claim that
    \[
    A=\big\{0,\,4,\,6,\,7,\,\dots,\,4^{m+1}\big\}.
    \]
	Indeed, the following decompositions cover all the required weights.
\begin{itemize}
\item If $x=0$, $x=4$, or $6\leq x\leq 4^m$, we may write $x=x+0+0+0$, so $x\in A$.
\item If $x=4^m+a$ with $1\leq a\leq 5$, we may write $x=6+\bigl(4^m+a-6\bigr)+0+0$, which shows that $x\in A$.
\item If $4^m+6\leq x\leq 2\cdot4^m$, we may write $x=4^m+\bigl(x-4^m\bigr)+0+0$, which shows that $x\in A$.
\item If $x=2\cdot4^m+a$ with $1\leq a\leq 5$, we may write $x=4^m+6+\bigl(4^m+a-6\bigr)+0$, which shows that $x\in A$.
\item If $2\cdot4^m+6\leq x\leq 3\cdot4^m$, we may write $x=4^m+4^m+\bigl(x-2\cdot4^m\bigr)+0$, which shows that $x\in A$.
\item If $x=3\cdot4^m+a$ with $1\leq a\leq 5$, we may write $x=4^m+4^m+6+\bigl(4^m+a-6\bigr)$, which shows that $x\in A$.
\item If $3\cdot4^m+6\leq x\leq 4^{m+1}$, we may write $x=4^m+4^m+4^m+\bigl(x-3\cdot4^m\bigr)$, which shows that $x\in A$.
\end{itemize}
	Here $m\geq2$ ensures that $6\leq4^m+a-6\leq4^m$ whenever $1\leq a\leq5$, so all the summands above belong to the defining set of $A$.
	Therefore, $\{0,4,6,7,\ldots,4^{m+1}\}\subseteq A$.
	Conversely, every element of $A$ is at most $4\cdot4^m=4^{m+1}$. Moreover, every positive summand is either $4$ or at least $6$, and a sum with at least two positive summands is at least $8$. Hence none of $1,2,3,5$ can belong to $A$. Thus $A\subseteq\{0,4,6,7,\ldots,4^{m+1}\}$.
	As a consequence, $A=\{0,4,6,7,\ldots,4^{m+1}\}$.
	
	From Proposition~\ref{prop:minimum-distance}, we have $d_4(3m,\,m+1)=4$, and from Proposition~\ref{prop:second-weight}, the second weight of $\mathrm{RM}_4(3m,\,m+1)$ is $6$. Hence the weights $1, 2, 3$ and $5$ cannot occur. Since the code has length $4^{m+1}$, no weight can exceed $4^{m+1}$. Together with the above inclusion, this yields
	\[
	\mathrm{wt}\big(\mathrm{RM}_4(3m,\,m+1)\big) = \big\{0,\,4,\,6,\,7,\,\dots,\,4^{m+1}\big\}.
	\]
	This completes the proof.
\end{proof}

\section{Quinary GRM codes}\label{sec:quinary}
\begin{prop}
	If $m\ge 2$, then
    \[
    \mathrm{wt}\big(\mathrm{RM}_5(4m-4,m)\big)=\{0,\,5,\,8,\,9,\,\dots,\,5^m\}.
    \]
\end{prop}
\begin{proof}
	We proceed by induction on $m$. A computation in Magma~\cite{magma} gives
	\[
\mathrm{wt}\big(\mathrm{RM}_5(4,2)\big) = \big\{0,\,5,\,8,\,9,\,\dots,\,25\big\}.
\]
	Thus the assertion holds for $m=2$.
	
	Now assume that the assertion holds for some $m\ge 2$. By Proposition~\ref{prop:spectrum-inclusion},
	\[
	\mathrm{wt}\big(\mathrm{RM}_5(4m,\,m+1)\big) \supseteq \bigoplus_{i=1}^5 \big\{0,\,5,\,8,\,9,\,\dots,\,5^m\big\}.
	\]
	Let
    \[
    A=\bigoplus_{i=1}^5 \big\{0,\,5,\,8,\,9,\,\dots,\,5^m\big\}.
    \]
    We claim that
    \[
    A=\big\{0,\,5,\,8,\,9,\,\dots,\,5^{m+1}\big\}.
    \]

	Indeed, the following decompositions cover all the required weights.
\begin{itemize}
\item If \(x=0\), \(x=5\), or \(8\leq x\leq 5^m\), we may write $x=x+0+0+0+0$, so \(x\in A\).
\item If \(x=5^m+a\) with \(1\leq a\leq 7\), we may write $x=8+\bigl(5^m+a-8\bigr)+0+0+0$, which shows that \(x\in A\).
\item If $5^m+8\leq x\leq 2\cdot 5^m$, we may write $x=5^m+\bigl(x-5^m\bigr)+0+0+0$, which shows that \(x\in A\).
\item	If \(x=2\cdot5^m+a\) with \(1\leq a\leq7\), we may write $x=5^m+8+\bigl(5^m+a-8\bigr)+0+0$, which shows that \(x\in A\).
\item If $2\cdot5^m+8\leq x\leq3\cdot5^m$, we may write $x=5^m+5^m+\bigl(x-2\cdot5^m\bigr)+0+0$, which shows that \(x\in A\).
\item If \(x=3\cdot5^m+a\) with \(1\leq a\leq7\), we may write $x=5^m+5^m+8+\bigl(5^m+a-8\bigr)+0$, which shows that \(x\in A\).
\item If $3\cdot5^m+8\leq x\leq4\cdot5^m$, we may write $x=5^m+5^m+5^m+\bigl(x-3\cdot5^m\bigr)+0$, which shows that \(x\in A\).
\item If \(x=4\cdot5^m+a\) with \(1\leq a\leq7\), we may write $x=5^m+5^m+5^m+8+\bigl(5^m+a-8\bigr)$, which shows that \(x\in A\).
\item If
	$4\cdot5^m+8\leq x\leq5^{m+1}$, we may write $x=5^m+5^m+5^m+5^m+\bigl(x-4\cdot5^m\bigr)$, which shows that \(x\in A\). 
\end{itemize}
    Here $8\leq 5^m+a-8\leq 5^m$ for $1\leq a\leq 7$ in the induction range, so every displayed summand belongs to the defining set of $A$.

Therefore, $\{0,5,8,9,\ldots,5^{m+1}\}\subseteq A$. Conversely, every element of \(A\) is at most $5\cdot5^m=5^{m+1}$. Moreover, since the smallest positive element of $\{0,5,8,9,\ldots,5^m\}$
	is \(5\), and every other positive element is at least \(8\), none of $1,2,3,4,6,7$ can belong to \(A\). Hence $A\subseteq\{0,5,8,9,\ldots,5^{m+1}\}$. As a consequence, $A=\{0,5,8,9,\ldots,5^{m+1}\}$.
	
	From Proposition~\ref{prop:minimum-distance}, we have $d_5(4m,\,m+1)=5$, and from Proposition~\ref{prop:second-weight}, the second weight of $\mathrm{RM}_5(4m,\,m+1)$ is $8$. Hence the weights $1, 2, 3, 4$ and $6, 7$ cannot occur. Since the code has length $5^{m+1}$, no weight can exceed $5^{m+1}$. Together with the above inclusion, this yields
	\[
	\mathrm{wt}\big(\mathrm{RM}_5(4m,\,m+1)\big) = \big\{0,\,5,\,8,\,9,\,\dots,\,5^{m+1}\big\}.
	\]
	This completes the proof.
\end{proof}
\section{GRM codes over \texorpdfstring{$\mathbb{F}_7$}{F7}}\label{sec:seven}
\begin{prop}
	If $m\ge 2$, then
    \[
    \mathrm{wt}\big(\mathrm{RM}_7(6m-4,m)\big)=\{0,\,5,\,6,\,7,\,\dots,\,7^m\}.
    \]
\end{prop}
\begin{proof}
	We proceed by induction on $m$. A computation in Magma~\cite{magma} gives
	\[
\mathrm{wt}\big(\mathrm{RM}_7(8,2)\big) = \big\{0,\,5,\,6,\,7,\,\dots,\,49\big\}.
\]
	Thus the assertion holds for $m=2$.
	
	Now assume that the assertion holds for some $m\ge 2$. By Proposition~\ref{prop:spectrum-inclusion},
	\[
	\mathrm{wt}\big(\mathrm{RM}_7(6m+2,\,m+1)\big) \supseteq \bigoplus_{i=1}^7 \big\{0,\,5,\,6,\,7,\,\dots,\,7^m\big\}.
	\]
	Let
    \[
    A=\bigoplus_{i=1}^7 \big\{0,\,5,\,6,\,7,\,\dots,\,7^m\big\}.
    \]
    We claim that
    \[
    A=\big\{0,\,5,\,6,\,7,\,\dots,\,7^{m+1}\big\}.
    \]
	Indeed, the following decompositions cover all the required weights. 
\begin{itemize}
\item If \(x=0\) or \(5\leq x\leq 7^m\), we may write $x=x+0+0+0+0+0+0$, so \(x\in A\).
\item	If \(x=7^m+a\) with \(1\leq a\leq4\), we may write $x=5+\bigl(7^m+a-5\bigr)+0+0+0+0+0$, which shows that \(x\in A\).
\item If $7^m+5\leq x\leq2\cdot7^m$, we may write $x=7^m+\bigl(x-7^m\bigr)+0+0+0+0+0$, which shows that \(x\in A\).
\item	If \(x=2\cdot7^m+a\) with \(1\leq a\leq4\), we may write $x=7^m+5+\bigl(7^m+a-5\bigr)+0+0+0+0$, which shows that \(x\in A\).
\item If $2\cdot7^m+5\leq x\leq3\cdot7^m$, we may write $x=7^m+7^m+\bigl(x-2\cdot7^m\bigr)+0+0+0+0$, which shows that \(x\in A\).
\item If \(x=3\cdot7^m+a\) with \(1\leq a\leq4\), we may write $x=7^m+7^m+5+\bigl(7^m+a-5\bigr)+0+0+0$, which shows that \(x\in A\).
\item If $3\cdot7^m+5\leq x\leq4\cdot7^m$, we may write $x=7^m+7^m+7^m+\bigl(x-3\cdot7^m\bigr)+0+0+0$, which shows that \(x\in A\).
\item	If \(x=4\cdot7^m+a\) with \(1\leq a\leq4\), we may write $x=7^m+7^m+7^m+5+\bigl(7^m+a-5\bigr)+0+0$, which shows that \(x\in A\).
\item If $4\cdot7^m+5\leq x\leq5\cdot7^m$, we may write $x=7^m+7^m+7^m+7^m+\bigl(x-4\cdot7^m\bigr)+0+0$, which shows that \(x\in A\).
	\item If \(x=5\cdot7^m+a\) with \(1\leq a\leq4\), we may write $x=7^m+7^m+7^m+7^m+5+\bigl(7^m+a-5\bigr)+0$, which shows that \(x\in A\).
\item	If $5\cdot7^m+5\leq x\leq6\cdot7^m$, we may write $x=7^m+7^m+7^m+7^m+7^m+\bigl(x-5\cdot7^m\bigr)+0$, which shows that \(x\in A\).
\item	If \(x=6\cdot7^m+a\) with \(1\leq a\leq4\), we may write $x=7^m+7^m+7^m+7^m+7^m+5+\bigl(7^m+a-5\bigr)$, which shows that \(x\in A\).
\item	If $6\cdot7^m+5\leq x\leq7^{m+1}$, we may write $x=7^m+7^m+7^m+7^m+7^m+7^m+\bigl(x-6\cdot7^m\bigr)$, which shows that \(x\in A\).

\end{itemize}
    Here $5\leq 7^m+a-5\leq 7^m$ for $1\leq a\leq 4$ in the induction range, so every displayed summand belongs to the defining set of $A$.
	Therefore, $\{0,5,6,7,\ldots,7^{m+1}\}\subseteq A$.
	Conversely, every element of \(A\) is at most $7\cdot7^m=7^{m+1}$.
	Moreover, since the smallest positive element of $\{0,5,6,7,\ldots,7^m\}$
	is \(5\), none of \(1,2,3,4\) can belong to \(A\). Hence
	$A\subseteq\{0,5,6,7,\ldots,7^{m+1}\}$.
	Consequently,
	$A=\{0,5,6,7,\ldots,7^{m+1}\}$.
	
	From Proposition~\ref{prop:minimum-distance}, we have $d_7(6m+2,\,m+1)=5$. Hence the weights $1, 2, 3, 4$ cannot occur. Since the code has length $7^{m+1}$, no weight can exceed $7^{m+1}$. Together with the above inclusion, this yields
	\[
	\mathrm{wt}\big(\mathrm{RM}_7(6m+2,\,m+1)\big) = \big\{0,\,5,\,6,\,7,\,\dots,\,7^{m+1}\big\}.
	\]
	This completes the proof.
\end{proof}

\section{Conclusion and open problems}\label{open problem}

We have determined the weight spectra of several families of
generalized Reed--Muller codes. In particular, our
results for $\mathrm{RM}_3(2m-c,m)$ with $c=3,4$, together with
the known spectrum for $c=2$, suggest that, for each fixed
$c\ge 2$ and all sufficiently large $m$, every integer from
$2d_c-3$ to $3^m$ occurs as a weight, where $d_c$ denotes the
minimum distance of $\mathrm{RM}_3(2m-c,m)$.
This motivates the following conjecture.

\begin{conj}\label{conj:ternary-spectrum}
	Let $c\ge 2$ be a fixed integer, and define
	\[
	d_c=
	\begin{cases}
		3^s, & \text{if } c=2s,\\
		2\cdot 3^s, & \text{if } c=2s+1.
	\end{cases}
	\]
	There exist an integer $m_0(c)\ge c$ and a set
	\[
	A_c\subseteq\{w\in\mathbb{Z}:d_c\le w<2d_c-3\},
	\]
	depending only on $c$, such that, for every $m\ge m_0(c)$,
	\[
	\operatorname{wt}\bigl(\mathrm{RM}_3(2m-c,m)\bigr)
	=
	\{0\}\cup A_c\cup\{2d_c-3,2d_c-2,\ldots,3^m\}.
	\]
\end{conj}

The conjecture holds for $c=2,3,4$, with
$A_2=\varnothing,
A_3=\{6,8\},
A_4=\{9,12\}.$
Determining the sets $A_c$ and establishing the conjecture for
$c\ge 5$ remains an open problem.

\end{document}